\documentclass[aps,pra,,groupedaddress]{revtex4-2}

\usepackage{amsthm,amsfonts,amsmath,amssymb,amscd, accents, mathrsfs, mathtools}
\usepackage{hyperref}
                                \usepackage{verbatim}
                                \usepackage{graphicx}

\swapnumbers
                                
\theoremstyle{plain}

\usepackage{enumitem}
\numberwithin{equation}{section}
\newtheorem{theorem}{Theorem}[section]

\newtheorem{proposition}[theorem]{Proposition}

\theoremstyle{definition}

\theoremstyle{remark}

\numberwithin{equation}{section}

\newcommand{\cR}{{\mathcal R}}

\newcommand{\cH}{{\mathcal H}}
\newcommand{\cL}{{\mathcal L}}

\newcommand{\cN}{{\mathcal N}}
\newcommand{\cM}{{\mathcal M}}
\newcommand{\cK}{{\mathcal K}}

\newcommand{\ket}[1]{\left\vert #1\right\rangle}
\newcommand{\bra}[1]{\left\langle #1\right\vert}

\newcommand{\Tr}{\mathrm{Tr}}

\newcommand{\supp}{\mathop{\mathrm{supp}}}

\newcommand{\be}{\begin{equation}}
\newcommand{\ee}{\end{equation}}
\newcommand{\bea}{\begin{eqnarray}}
\newcommand{\eea}{\end{eqnarray}}
\newcommand{\beann}{\begin{eqnarray*}}
\newcommand{\eeann}{\end{eqnarray*}}

\usepackage{color}

\usepackage{mathtools}

\usepackage{hyperref}
\usepackage{cleveref}
\usepackage{autonum}

\begin{document}

\title{Continuity inequalities for sandwiched R\'enyi and Tsallis conditional entropies with application to the channel entropy continuity}
\author{Anna Vershynina}
\affiliation{\small{Department of Mathematics, Philip Guthrie Hoffman Hall, University of Houston, 
3551 Cullen Blvd., Houston, TX 77204-3008, USA}}

\date{\today}

\begin{abstract} For the sandwiched R\'enyi entropy the conditional entropy can be defined two ways: $\tilde{H}^\downarrow_\alpha(A|B)_\rho, \tilde{H}^\uparrow_\alpha(A|B)_\rho$. In the limiting case, $\alpha=1$, both definitions consolidate into conditional entropy $H(A|B)=S(AB)-S(B)$. The continuity inequality for conditional entropy $H(A|B)$, called the Alicki-Fannes-Winter (AWF) inequality, shows that if the states are close in trace-distance, then the conditional entropies are also close. Having the AWF inequality for conditional entropy, we show that the channel entropy defined through the relative entropy is continuous with respect to the diamond-distance between channels. Inspired by this, similar continuity inequalities for the R\'enyi conditional entropy $\tilde{H}^\uparrow_\alpha$ were presented in \cite{marwah2022uniform}. We provide continuity bounds for the sandwiched R\'enyi and Tsallis conditional entropies $\tilde{H}^\downarrow_\alpha(A|B)_\rho, \tilde{T}^\downarrow_\alpha(A|B)_\rho$ for states with the same marginal on the conditioning system. Similar to the previous bounds, our bound depends only on the dimension of the conditioning system. We apply this result to prove continuity of the channel entropy for R\'enyi and Tsallis channel entropies defined through the sandwiched R\'enyi and Tsallis relative entropies.
\end{abstract}
\maketitle

\section{Introduction}

For quantum states a well-known continuity bound of the quantum entropy is given by the Fannes-Audenaert inequality \cite{alicki2004continuity,  audenaert2007sharp,  Fan73, winter2016tight}. The inequality provides an upper bound on the entropy difference in terms of the trace-distance:
$$|S(\rho)-S(\sigma)|\leq T \log [d-1]+s_2(T)\ , $$
where $T=\frac{1}{2}\|\rho-\sigma\|_1$, $d$ is the Hilbert space dimension, and $s_2(p)=-p\log p-(1-p)\log(1-p)$ is the binary entropy. Various proofs and generalization were found afterwards \cite{alicki2004continuity, AE11, audenaert2024continuity, P08, winter2016tight}. These bounds are applied, in particular, to entanglement measures \cite{nielsen2000continuity}, the capacity of quantum channels \cite{LeuS09, Sir17}, and others.

We ask a similar continuity question for the entropy of a channel. The channel entropy is defined through the relative entropy of channels \cite{CMM18}: for two quantum channels $\cN_{A\rightarrow B}$ and $\cM_{A\rightarrow B}$, the relative entropy between them is defined as 
\begin{equation}\label{def:qre-channel}
D(\cN \| \cM)=\sup_{\rho_{AR}} D(\cN\otimes I(\rho)\| \cM\otimes I(\rho))\ .
\end{equation}
Here the supremum is taken over all systems $R$ of any dimension and all states $\rho_{AR}$. However, it is sufficient to consider only pure states $\rho_{AR}$ with system $R$ being isomorphic to system $A$, because of the state purification, the data-processing inequality, and the Schmidt decomposition theorem. This definition was generalized in \cite{LKDW18} by taking any generalized divergence instead of the relative entropy (i.e. satisfying the data processing inequality), leading to a {divergence of channels}. When divergence is a trace-distance \cite{Wat18}, then the divergence of channels is called a {diamond-distance of channels}.

The {entropy of a channel} $\cN$ is then defined as \cite{Y18}
\begin{equation}\label{intro-def:entropy-channel}
S(\cN)=\log|B|-D(\cN\|\cR)\ ,
\end{equation}
here $D$ is the relative entropy of the channels,  $\cR_{A\rightarrow B}(\rho_A)=\Tr(\rho_A)\pi_B$ is a completely randomizing/depolarizing channel, and $\pi_B=I_B/|B|$ is the maximally mixed state. This definition generalizes the static case, when the entropy of a state can be written as $S(\rho)=\log d-D(\rho\| \pi).$

In Section \ref{sec:Rel-entropy} we show that if two quantum channels are close to each other in diamond-distance, then their channel entropy is also close
\begin{equation}\label{intro-eq:channel-entropy-cont}
|S(\cN)-S(\cM)|\leq f(\epsilon,|B|)\ , 
\end{equation}
here $f(\epsilon, |B|)$ is the upper bound discussed below in (\ref{eq:intro}). The proof relies on the fact that the channel entropy $S(\cN)$ can be written in terms of the conditional entropy $H(A|B)$. And the continuity inequality for the channel entropy reduces to a continuity inequality for the conditional entropy. 
For the conditional entropy, $H(A|B)=S(AB)-S(B)$, a continuity inequality was first proved by Alicki and Fannes \cite{alicki2004continuity} and later improved by Winter \cite{winter2016tight}. The inequality is now known as the Alicki-Fannes-Winter (AFW) inequality
\begin{equation}\label{eq:intro}
|H(A|B)_\rho-H(A|B)_\sigma|\leq 2\epsilon \log|A|+(\epsilon+1)\log(\epsilon+1)-\epsilon\log\epsilon=:f(\epsilon, |A|) \ ,
\end{equation}
where the trace-distance $T\leq \epsilon\in[0,1]$. Various analogues and generalizations of this inequality have since been obtained in the literature. Either of these can be used in place of the AFW inequality to derive the upper bound in the continuity of the channel entropy inequality (\ref{intro-eq:channel-entropy-cont}). In particular, the authors in \cite{berta2025continuity} proved that for states $\rho$ and $\sigma$ with the same marginal $\rho_B=\sigma_B$, for $\epsilon$ close to $1$,  it holds that
\begin{equation}\
|H(A|B)_\rho-H(A|B)_\sigma|\leq \epsilon \log(|A|\cdot \text{SN}(\rho_{AB}))+\epsilon\log\epsilon+ (1-\epsilon)\log(1-\epsilon)\ ,
\end{equation}
where $\text{SN}(\rho)$ is the Schmidt number of $\rho$, which is upper bounded by $\min\{|A|, |B|\}$. The equal-marginal setting is precisely the case needed for the continuity of the entropy itself, and therefore is a natural regime for continuity bounds.

Note that the conditional entropy can be written in three equivalent ways
\begin{align}
H(A|B)_\rho&=S(\rho_{AB})-S(\rho_B)\label{eq:intro-conditional-1} \\
&=-D(\rho_{AB}\|I_A\otimes \rho_B)\label{eq:intro-conditional-2}\\
&=-\min_{\sigma_B}D(\rho_{AB}\|I_A\otimes \sigma_B)\ .\label{eq:intro-conditional-3}
\end{align}

For the R\'enyi entropy, it was shown that the generalization of the first expression (\ref{eq:intro-conditional-1}) has severe limitations, as it does not satisfy the data processing inequality \cite{tomamichel2014relating}. The other two expressions, however, give a very useful generalizations, studied in particular in \cite{arimoto1977information, gallager1979source, hayashi2014large, hayashi2013security, marwah2022uniform, muller2013quantum, tomamichel2009fully, tomamichel2014relating, yagi2012finite}. Here, we focus on the last two expressions for the sandwiched R\'enyi relative entropy $\tilde{D}_\alpha$ and the sandwiched Tsallis relative entropy $\tilde{D}^T_\alpha$.

Conditional sandwiched R\'enyi entropies for a state $\rho_{AB}$ are then defined as
\begin{align}
\tilde{H}^\downarrow_\alpha(A|B)_\rho&=-\tilde{D}_\alpha(\rho_{AB}\|I_A\otimes \rho_B)=\frac{1}{1-\alpha}\log \Tr\{(\rho_B^{\frac{1-\alpha}{2\alpha}}\rho\rho_B^{\frac{1-\alpha}{2\alpha}})^\alpha\}\ ,\\
\tilde{H}^\uparrow_\alpha(A|B)_\rho&=-\min_{\sigma_B}\tilde{D}_\alpha(\rho_{AB}\|I_A\otimes \sigma_B)\ .
\end{align}

Continuity inequality for the conditional R\'enyi entropy  $\tilde{H}^\uparrow_\alpha$ was derived in \cite{marwah2022uniform}: if $\frac{1}{2}\|\rho_{AB}-\sigma_{AB}\|_1\leq \epsilon$ for $\epsilon\in[0,1]$ then
\begin{equation}
|\tilde{H}^\uparrow_\alpha(A|B)_\rho- \tilde{H}^\uparrow_\alpha(A|B)_\sigma|\leq f^\uparrow_{\alpha, |A|}(\epsilon)\ ,
\end{equation}
where 
$$ f^\uparrow_{\alpha, d} (\epsilon)=\begin{cases}
\log(1+\epsilon)+\frac{1}{1-\alpha}\log\Bigg(\epsilon^\alpha d^{2(1-\alpha)}+1-\frac{\epsilon}{(1+\epsilon)^{1-\alpha}}\Bigg)\ , \ \ \ \alpha<1\ ,\\
\log(1+\sqrt{2\epsilon})+\frac{1}{1-\beta}\log\Bigg(\sqrt{2\epsilon}^\beta d^{2(1-\beta)}+1-\frac{\sqrt{2\epsilon}}{(1+\sqrt{2\epsilon})^{1-\beta}}\Bigg)\ , \ \ \ \alpha>1 \ .
\end{cases}$$
Here $\alpha^{-1}+\beta^{-1}=2$. The last case, when $\alpha>1$, is derived from the first case, when $\alpha<1$, because of the duality property: for a pure state $\rho_{ABC}$, we have $\tilde{H}^\uparrow_\alpha(A|B)_\rho+\tilde{H}^\uparrow_\beta(A|C)_\rho=0, $
for $\alpha^{-1}+\beta^{-1}=2$. Note that there is no similar duality inequality for $\tilde{H}^\downarrow_\alpha$, so each case for $\alpha$ must be considered separately.

In Section \ref{sec:Renyi-sand}, we prove the continuity of the conditional entropy $\tilde{H}^\downarrow_\alpha$:  if $\frac{1}{2}\|\rho_{AB}-\sigma_{AB}\|_1\leq \epsilon$ for $\epsilon\in[0,1]$ and $\rho_B=\sigma_B$, then 
$$|\tilde{H}_\alpha^\downarrow(A|B)_\rho-\tilde{H}_\alpha^\downarrow(A|B)_\sigma|\leq f_{\alpha, |A|}(\epsilon)\ , $$
where 
$$f_{\alpha,d}(\epsilon)=\begin{cases}
\log(1+\epsilon)+\frac{1}{1-\alpha}\log\Bigg(1+\epsilon^\alpha d^{2(1-\alpha)}\Bigg)\ , \ \ \ \alpha\in[\frac{1}{2}, 1)\\
\frac{\alpha}{\alpha-1}\log(1+\epsilon)\ , \ \ \ \alpha>1 \ .
\end{cases}
$$
Note that for every fixed $\alpha$, the $f_{\alpha,d}(\epsilon)\rightarrow 0$ as $\epsilon\rightarrow 0$.

Using this inequality, we prove the continuity of the R\'enyi channel entropy defined through the sandwiched R\'enyi relative entropy: if two quantum channels are close in diamond-distance, $\frac{1}{2}\|\cN-\cM\|_\diamond\leq\epsilon$, then the R\'enyi channel entropy for $\alpha\in[\frac{1}{2}, 1)\cup(1,\infty)$ is also close
$$|\tilde{S}_\alpha(\cN)-\tilde{S}_\alpha(\cM)|\leq f_{\alpha,|B|}(\epsilon)\ .$$

In Section \ref{sec:Tsallis-sand}, we provide similar continuity inequalities for the sandwiched Tsallis conditional entropy and the Tsallis channel entropy defined in terms of the sandwiched Tsallis relative entropy for $\alpha\in[\frac{1}{2},1)\cup(1,2)$. The sandwiched Tsallis relative entropy is defined as
$$
\tilde{D}^T_\alpha(\rho\|\sigma)=\frac{1}{\alpha-1}\Bigg( \Tr\{(\sigma^{\frac{1-\alpha}{2\alpha}}\rho\sigma^{\frac{1-\alpha}{2\alpha}})^\alpha\}-1\Bigg)\ ,
$$
for $\alpha<1$, or $\alpha>1$ and $\supp\rho\subseteq\supp\sigma$.

Conditional sandwiched Tsallis entropies for a state $\rho_{AB}$ are defined as
\begin{align}
\tilde{T}^\downarrow_\alpha(A|B)_\rho&=-\tilde{D}^T_\alpha(\rho_{AB}\|I_A\otimes \rho_B)=\frac{1}{1-\alpha}\Bigg( \Tr\{(\rho_B^{\frac{1-\alpha}{2\alpha}}\rho\rho_B^{\frac{1-\alpha}{2\alpha}})^\alpha\}-1\Bigg)\ ,\\
\tilde{T}^\uparrow_\alpha(A|B)_\rho&=-\min_{\sigma_B}\tilde{D}^T_\alpha(\rho_{AB}\|I_A\otimes \sigma_B)\ .
\end{align}

Similarly to the sandwiched R\'enyi conditional entropy we have the following continuity inequality. For states with the same marginals $\rho_B=\sigma_B$ such that $\frac{1}{2}\|\rho-\sigma\|_1=\epsilon\in[0,1]$, we have
$$
|\tilde{T}_\alpha^\downarrow(A|B)_\rho-\tilde{T}_\alpha^\downarrow(A|B)_\sigma|\leq  f^T_{\alpha,|A|}(\epsilon)\ ,
$$
where
$$f^T_{\alpha,d}(\epsilon)=\begin{cases}
\frac{1}{1-\alpha}((1+\epsilon^\alpha )(1+\epsilon)^{1-\alpha}-1)d^{1-\alpha}\ , \ \ \ \alpha\in[\frac{1}{2}, 1)\\
\frac{1}{\alpha-1} \Bigg[\Bigg((1+\epsilon)^{\alpha-1}-1\Bigg)d^{\alpha-1}+\epsilon(1+\epsilon)^{\alpha-1}d^{1-\alpha}\Bigg]\ , \ \ \ \alpha\in(1,2)\ .
\end{cases}
$$
Note that for every fixed $\alpha$, the $f^T_{\alpha,d}(\epsilon)\rightarrow 0$ as $\epsilon\rightarrow 0$.

We define the $\alpha$-Tsallis channel entropy as follows
$$\tilde{S}^T_\alpha(\cN)=\frac{|B|^{1-\alpha}-1}{1-\alpha}-|B|^{1-\alpha}\tilde{D}^T_\alpha(\cN\|\cR)\ .
$$
This definition is different from the form (\ref{intro-def:entropy-channel}) for the relative entropy and the R\'enyi relative entropies. It is because the Tsallis relative entropy has a different scaling: $\tilde{D}^T_\alpha(\rho\|c\sigma)=\frac{c^{1-\alpha}-1}{\alpha-1}+c^{1-\alpha}\tilde{D}^T_\alpha(\rho\|\sigma)$, 
 resulting in a different form for the entropy of a state in terms of the relative entropy.

This Tsallis channel entropy is monotone under uniformity preserving superchannels, normalized, bounded (Theorem \ref{thm:Tsallis-bound}), and pseudo-additive. Note that to show boundedness, we used the conditional Tsallis entropy ${T}_\alpha^\downarrow$ defined through the Tsallis relative entropy (non-sandwiched). In Theorem \ref{thm:Tsallis-additive}, we show that the pseudo-additivity takes the form
$$\tilde{S}_\alpha^T(\cN\otimes\cM)=\tilde{S}_\alpha^T(\cN)+\tilde{S}_\alpha^T(\cM)+(1-\alpha)\tilde{S}_\alpha^T(\cN)\tilde{S}_\alpha^T(\cM)\ .$$

In Theorem \ref{thm:Tsallis-continuity-entropy}, we show the continuity of the Tsallis channel entropy: for two quantum channels close to each other in diamond-distance, $\frac{1}{2}\|\cN-\cM\|_\diamond\leq\epsilon$, we have
$$|\tilde{S}^T_\alpha(\cN)-\tilde{S}^T_\alpha(\cM)|\leq f^T_{\alpha, |B|}(\epsilon)\ . $$

\section{Preliminaries}
\subsection{Definitions}

In this paper all Hilbert spaces are finite dimensional. We denote them as $\cH_A, \cH_B, \dots$, where subscripts indicate the corresponding system. A Hilbert space corresponding to multiple systems, e.g. AB, is a tensor product of individual subsystems, e.g. $\cH_A\otimes\cH_B$.  For a Hilbert space $\cH_A$, its dimension is denoted as $|A|:=\dim\cH_A$.  For a Hilbert space $\cH$, we denote $\cL(\cH)$ the space of all  linear operators on $\cH$.

A quantum state or a density operator $\rho_A\in\cL(\cH_A)$ on a Hilbert space $\cH_A$ is a positive semidefinite, trace-normalized operator, i.e. $\rho_A\geq 0$ and $\Tr\,\rho_A=1$. A state is pure if it is rank-one. A pure state $\psi_A$ has an associated  vector $\ket{\psi}_A\in\cH_A$ such that $\langle{\psi}|\psi\rangle=1$ and $\psi_A=\ket{\psi}\bra{\psi}_A$. 
 
A quantum channel $\cN:A\rightarrow B$ is a linear completely-positive trace-preserving (CPTP) map from $\cL(\cH_A)$ to $\cL(\cH_B)$. The channel can also be denoted as $\cN_{A\rightarrow B}$. The identity channel on system $A$ is denoted as $I_A$. The subscript in the channels is dropped if it is evident which systems are involved. 

Completely depolarizing/randomizing channel is defined as
$$\cR_{A\rightarrow B}(\rho_A)=\Tr(\rho_A)\pi_B\ , $$
where $\pi_B=I_B/|B|$ is the maximally mixed state.

A {superchannel} \cite{ChDAP08} $\Lambda$ transforms a quantum channel $\cN_{A\rightarrow B} $ to a channel from $C$ to $D$ as follows
\begin{equation}\label{def:superch}
\Lambda(\cN_{A\rightarrow B})_{C\rightarrow D}=\cM_{BE\rightarrow D}\circ(\cN_{A\rightarrow B}\otimes I_E)\circ \cK_{C\rightarrow AE}\ , 
\end{equation}
with the ancillary system $E$, and channels $\cM_{BE\rightarrow D}$ and $\cK_{C\rightarrow AE}$.

A uniformity preserving superchannel $\Lambda$ is the one sending a completely randomizing channel to a completely randomizing one, i.e. $\Lambda(\cR_{A\rightarrow B})=\cR_{C\rightarrow D}$ with $|A|=|C|$ and $|B|=|D|$.

The quantum (Umegaki) relative entropy \cite{U62} is defined as $D(\rho\|\sigma)=\Tr(\rho\log\rho-\rho\log\sigma)$ when $\supp\rho\subseteq\supp\sigma$, and $+\infty$ otherwise. The quantum entropy of a state $\rho$ is defined as $S(\rho)=-\Tr(\rho\log\rho)$. The quantum conditional entropy can be defined any one of the following three ways:
\begin{align}
H(A|B)_\rho&=S(\rho_{AB})-S(\rho_B) \\
&=-D(\rho_{AB}\|I_A\otimes \rho_B)\\
&=-\min_{\sigma_B}D(\rho_{AB}\|I_A\otimes \sigma_B)\ .
\end{align}

For two quantum channels $\cN_{A\rightarrow B}$ and $\cM_{A\rightarrow B}$, the relative entropy of channels \cite{CMM18}  is defined as
\begin{equation}
D(\cN\|\cM)=\sup_{\rho_{AR}} D(\cN\otimes I(\rho)\| \cM\otimes I(\rho))\ .
\end{equation}
Here the supremum is taken over all systems $R$ of any dimension and all states $\rho_{AR}$. However, it is sufficient to consider only pure states $\rho_{AR}$ with system $R$ being isomorphic to system $A$, because of the state purification, the data-processing inequality, and the Schmidt decomposition theorem. 

Taking other divergences instead of the relative entropy $D(\cdot \| \cdot)$ above, results in various relative entropies of channels, e.g. R\'enyi and Tsallis relative entropies.

Trace-norm of a linear map $X$ is defined as $\|X\|_1=\Tr\sqrt{X^*X}$. Then the trace-distance between two states $\rho$ and $\sigma$ is defined as $\|\rho-\sigma\|_1=\Tr|\rho-\sigma|$. Sometimes a factor of $\frac{1}{2}$ is added in the definition of a trace-distance. Will use these notions interchangeably, as it will be clear whether or not a factor is present or it makes no difference.

The {trace-distance of quantum channels} $\cN, \cM:A \rightarrow B$ is defined as 
\begin{equation}\label{def:trace-norm-channel}
\|\cN -\cM\|_1=\sup_{\rho_A} \|\cN(\rho) - \cM (\rho)\|_1\ .
\end{equation}

The {diamond-distance of channels} $\cN, \cM:A \rightarrow B$ is defined as 
\begin{equation}\label{def:qre-channel}
\|\cN - \cM\|_\diamond=\sup_{\rho_{AR}} \|\cN_{A\rightarrow B}\otimes I_R(\rho)- \cM_{A\rightarrow B}\otimes I_R(\rho)\|_1\ ,
\end{equation}
where, similarly to the above, it is enough to consider $\dim R=\dim A$ and only pure states in the maximization.

\subsection{Relative entropy}

The {\bf entropy of a quantum channel} \cite{Y18} $\cN_{A\rightarrow B}$ is defined as
\begin{equation}\label{def:entropy-channel}
S(\cN)=\log|B|-D(\cN\|\cR)\ ,
\end{equation}
where $D(\cdot\|\cdot)$ is based on the Umegaki relative entropy $D(\rho\|\sigma)=\Tr(\rho\log\rho-\rho\log\sigma).$

Note that
$$D(\cN\|\cR)=\sup_\psi D(\cN_{A\rightarrow B}\otimes I_R \ket{\psi}\bra{\psi}_{AR} \|\cR_{A\rightarrow B}\otimes I_R\ket{\psi}\bra{\psi}_{AR})=\sup_\psi D(\rho_{BR}\|\pi_B\otimes\rho_R)\ ,
 $$
 where $\rho_{BR}=\cN\otimes I\ket{\psi}\bra{\psi}$.

The entropy of a quantum channel has the following properties \cite{G19, GW21, Y18}:
\begin{itemize}
\item (Monotonicity) Since the divergence is monotone under quantum channels, the generalized channel entropy is monotone under uniformity preserving superchannels: For any uniformity preserving superchannel $\Lambda$ (i.e. sending a completely randomizing channel to a completely randomizing one, $\Lambda(\cR_{A\rightarrow B})=\cR_{C\rightarrow D}$ with $|A|=|C|$ and $|B|=|D|$), we have
$$S(\Lambda(\cN))\geq S(\cN)\ .$$

\item (Normalization) By definition, the entropy of a completely randomizing channel $\cR$ is $S(\cR)=\log|B|$.

Let $\Phi_\sigma(\rho_A)=\sigma_B$ be a replacer channel for some fixed state $\sigma$. Then
\begin{align}
D(\Phi_\sigma\otimes I \ket{\psi}\bra{\psi}_{AR} \|\cR\otimes I\ket{\psi}\bra{\psi}_{AR})&=D(\sigma_B\otimes \psi_R\|\pi_B\otimes\psi_R)\\
&=D(\sigma\|\pi)\ .
\end{align}
Here we used the stability property of the divergence, which is a consequence of the monotonicity property. 

Therefore, the entropy of the replacer channel is
\begin{align}
S(\Phi_\sigma)&=\log|B|-D(\cN\|\cR)\\
&=\log|B|-D(\sigma\|\pi)\\
&=S(\sigma)\ .
\end{align}

And for a replacer channel that replaces any state with a pure state, the entropy of this channel is zero,  i.e. $S(\Phi_\phi)=0$ for $\Phi_\phi(\rho)=\ket{\phi}\bra{\phi}$ for some fixed pure state $\ket{\phi}$.

\item (Additivity) For any two quantum channels, $S(\cN\otimes \cM)=S(\cN)+S(\cM)$.
\item (Boundedness) The entropy of a channel could be negative, but it is bounded, $|S(\cN)|\leq \log|B|$. The lowest value is achieved for the identity channel, and the highest value is achieved for a completely randomizing channel. 
\end{itemize}

\section{Continuity Inequalities}

\subsection{Relative entropy}\label{sec:Rel-entropy}

The channel entropy can be written as a infimum of a conditional entropy:
\begin{align}
S(\cN)&=\log|B|-D(\cN\|\cR)\\
&=\log|B|-\sup_{\psi} D(\cN_{A\rightarrow B}\otimes I_R\ket{\psi}\bra{\psi}_{AR}\|\cR_{A\rightarrow B}\otimes I_R \ket{\psi}\bra{\psi}_{AR})\\
&=\log|B|-\sup_{\psi} D(\cN\otimes I\ket{\psi}\bra{\psi}\|\pi_B\otimes \psi_R)\\
&=-\sup_{\psi} D(\cN\otimes I\ket{\psi}\bra{\psi}\|I_B\otimes \psi_R)\\
&=\inf_{\psi} H(B|R)_{\cN\otimes I \ket{\psi}\bra{\psi}}.
\end{align}
Here $\pi_B=I_B/|B|$ and $ \psi_R=\Tr_A\ket{\psi}\bra{\psi}$. We used $D(\rho\|c\sigma)=D(\rho\|\sigma)-\log c$ and $H(B|R)_\rho=-D(\rho_{BR}\|I_B\otimes \rho_R)$.

Recall the AFW continuity inequality for the conditional entropy \cite{alicki2004continuity, winter2016tight}: Let $\rho_{AB}, \sigma_{AB}$ be two states such that $\frac{1}{2}\|\rho-\sigma\|_1\leq \epsilon$ for $\epsilon\in[0,1]$. Then 
\begin{equation}
|H(A|B)_\rho-H(A|B)_\sigma|\leq 2\epsilon \log|A|+(\epsilon+1)\log(\epsilon+1)-\epsilon\log\epsilon=:f(\epsilon, |A|) \ .
\end{equation}
One may use any other valid upper bound above, which will appear in the continuity theorem below.

\begin{theorem}\label{thm:cont} (Continuity of channel entropy) Let $\cN$ and $\cM$ be two channels from $A$ to $B$ such that $\frac{1}{2}\|\cN-\cM\|_\diamond\leq\epsilon$. Then
$$|S(\cN)-S(\cM)|\leq f(\epsilon,|B|)\ , $$
where $f(\epsilon,|B|)=2\epsilon \log|B|+(\epsilon+1)\log(\epsilon+1)-\epsilon\log\epsilon$.
\end{theorem}
\begin{proof}
WLOG suppose that $S(\cN)\geq S(\cM)$.  

Since $S(\cM)=\inf_\psi H(B|R)_{\cM\otimes I \ket{\psi}\bra{\psi}}$, for any $\delta>0$, there exists a state $\omega$ such that
$$H(B|R)_{\cM\otimes I \ket{\omega}\bra{\omega}}<S(\cM)+\delta \ . $$
Then, since $S(\cN)=\inf_\psi H(B|R)_{\cM\otimes I \ket{\psi}\bra{\psi}}\leq H(B|R)_{\cN\otimes I\ket{\omega}\bra{\omega}}$, we have
\begin{align}
S(\cN)-S(\cM)&< H(B|R)_{\cN\otimes I\ket{\omega}\bra{\omega}}- H(B|R)_{\cM\otimes I \ket{\omega}\bra{\omega}}+\delta\\
&\leq f(\epsilon,|B|)+\delta\ ,
\end{align}
since $\|\cN\otimes I\ket{\omega}\bra{\omega}-\cM\otimes I\ket{\omega}\bra{\omega}\|_1\leq\|\cN-\cM\|_\diamond=\sup_\rho \|\cN\otimes I(\rho)-\cM\otimes I(\rho) \|_1\leq 2\epsilon$ and the AWF inequality.

Taking $\delta\rightarrow 0$, we obtain the necessary continuity inequality.

\end{proof}

\subsection{Sandwiched R\'enyi relative entropy}\label{sec:Renyi-sand}

R\'enyi  entropy is defined as, for $\alpha>0$ and $\alpha\neq 1$,
\begin{equation}
S_\alpha(\rho)=\frac{1}{1-\alpha}\log \Tr\{\rho^\alpha\}\ .
\end{equation}

The sandwiched R\'enyi relative entropy is defined as
\begin{equation}
\tilde{D}_\alpha(\rho\|\sigma)=\frac{1}{\alpha-1}\log \Tr\{(\sigma^{\frac{1-\alpha}{2\alpha}}\rho\sigma^{\frac{1-\alpha}{2\alpha}})^\alpha\}\ ,
\end{equation}
for $\alpha<1$, or $\alpha>1$ and $\supp\rho\subseteq\supp\sigma$.

Note that the sandwiched R\'enyi relative entropy obeys the data processing inequality for $\alpha\in[\frac{1}{2}, 1)\cup (1,\infty)$ and is jointly convex for $\alpha\in[\frac{1}{2}, 1)$. Moreover, the functional $(\rho,\sigma)\mapsto \Tr\{(\sigma^{\frac{1-\alpha}{2\alpha}}\rho\sigma^{\frac{1-\alpha}{2\alpha}})^\alpha\}$ is jointly concave for $\alpha\in[\frac{1}{2}, 1)$ and jointly convex for $\alpha>1$, see \cite{frank2013monotonicity}.

Conditional sandwiched R\'enyi entropies for a state $\rho_{AB}$ are defined as
\begin{align}
\tilde{H}^\downarrow_\alpha(A|B)_\rho&=-\tilde{D}_\alpha(\rho_{AB}\|I_A\otimes \rho_B)=\frac{1}{1-\alpha}\log \Tr\{(\rho_B^{\frac{1-\alpha}{2\alpha}}\rho\rho_B^{\frac{1-\alpha}{2\alpha}})^\alpha\}\ ,\\
\tilde{H}^\uparrow_\alpha(A|B)_\rho&=-\min_{\sigma_B}\tilde{D}_\alpha(\rho_{AB}\|I_A\otimes \sigma_B)\ .
\end{align}

A number of inequalities relating and bounding these and other conditional entropies was presented in a unified form in \cite{zhu2017coherence}, some original and some from other sources \cite{beigi2013sandwiched, hayashi2017quantum, leditzky2017data, muller2013quantum, tomamichel2014relating}. In particular, we will be using the boundedness of both conditional entropies: for any state $\rho_{AB}$,
\begin{align}
|\tilde{H}^\downarrow_\alpha(A|B)_\rho|&\leq \log|A|\ ,\label{eq:down-log} \\
|\tilde{H}^\uparrow_\alpha(A|B)_\rho|&\leq \log|A|\ .\label{eq:up-log}
\end{align}

\begin{theorem}\label{thm:sand-R-cond}
Let $\alpha\in[\frac{1}{2}, 1)$. Suppose that the states $\rho_{AB}$ and $\sigma_{AB}$ have the same marginals $\rho_B=\sigma_B$ and they are close to each other in trace-distance $\frac{1}{2}\|\rho-\sigma\|_1=\epsilon\in[0,1]$. Then
\begin{equation}
|\tilde{H}_\alpha^\downarrow(A|B)_\rho-\tilde{H}_\alpha^\downarrow(A|B)_\sigma|\leq \log(1+\epsilon)+\frac{1}{1-\alpha}\log\Bigg(1+\epsilon^\alpha |A|^{2(1-\alpha)}\Bigg)\ .
\end{equation}
\end{theorem}

\begin{proof}
When $\epsilon=0$, the bound is trivial. Then suppose that $\epsilon>0$. 

Denote $\omega_{AB}=I_A\otimes \rho_B=I_A\otimes\sigma_B$ and  $\gamma=\frac{1-\alpha}{2\alpha}$. Then for $\delta=\rho$ or $\delta=\sigma$ we have $\tilde{H}^\downarrow_\alpha(A|B)_\delta=-\tilde{D}_\alpha(\delta_{AB}\|\omega_{AB})=\frac{1}{1-\alpha}\log\Tr\{ (\omega^\gamma\delta\omega^\gamma)^\alpha\}$. Therefore,
\begin{align}
\Tr\{ (\omega^\gamma\rho\omega^\gamma)^\alpha\}&=2^{(1-\alpha)\tilde{H}^\downarrow_\alpha(A|B)_\rho}\ ,\label{eq:max-exp-rho}\\
\Tr\{ (\omega^\gamma\sigma\omega^\gamma)^\alpha\}&=2^{(1-\alpha)\tilde{H}^\downarrow_\alpha(A|B)_\sigma}\ . \label{eq:max-exp-sigma}
\end{align}

Moreover, for any state $\delta_{AB}$,
\begin{align}
\tilde{H}^\uparrow_\alpha(A|B)_\delta&=-\min_{\xi_B}\tilde{D}_\alpha(\delta_{AB}\|I_A\otimes \xi_B)\\
&=\max_{\xi_B}\frac{1}{1-\alpha}\log \Tr\{(\xi_B^{\frac{1-\alpha}{2\alpha}}\delta\xi_B^{\frac{1-\alpha}{2\alpha}})^\alpha\}\\
&=\frac{1}{1-\alpha}\log\max_{\xi_B} \Tr\{(\xi_B^{\frac{1-\alpha}{2\alpha}}\delta\xi_B^{\frac{1-\alpha}{2\alpha}})^\alpha\}\ .
\end{align}
Therefore, for any state $\delta_{AB}$,
\begin{equation}
\max_{\xi_B}\Tr\{(\xi_B^{\frac{1-\alpha}{2\alpha}}\delta\xi_B^{\frac{1-\alpha}{2\alpha}})^\alpha\}=2^{(1-\alpha)\tilde{H}^\uparrow_\alpha(A|B)_\delta}\ .
\end{equation}
And, in particular, 
\begin{equation}\label{eq:up-exp}
\Tr\{ (\omega^\gamma\delta\omega^\gamma)^\alpha\}\leq 2^{(1-\alpha)\tilde{H}^\uparrow_\alpha(A|B)_\delta}\ .
\end{equation}

 Let us decompose $\rho-\sigma=P'-Q'$ into positive $P'\geq0$ and negative $Q'\geq 0$ commuting parts. Then $\Tr P'=\Tr Q'=\epsilon$. Denote $P=P'/\epsilon$ and $Q=Q'/\epsilon$. Then $P, Q$ are density operators.

Denote 
\begin{equation}\label{eq:Delta}
\Delta_{AB}:=\frac{1}{1+\epsilon}\rho+\frac{\epsilon}{1+\epsilon}Q=\frac{1}{1+\epsilon}\sigma+\frac{\epsilon}{1+\epsilon}P\ .
\end{equation}
Recall McCarthy's inequality \cite{mccarthycp} or Rotfel'd inequality \cite{Rotfeld1969}: for $X, Y \geq 0$ and $\alpha\in[0,1]$, we have
\begin{equation}
\Tr \{(X+Y)^\alpha\}\leq \Tr\{X^\alpha\}+\Tr\{Y^\alpha\}\ .
\end{equation}
Taking $X=\frac{1}{1+\epsilon}\omega^\gamma\rho\omega^\gamma$ and $Y=\frac{\epsilon}{1+\epsilon}\omega^\gamma Q\omega^\gamma$ in the McCarthy's inequality, we have
\begin{align}
\Tr\{(\omega^\gamma\Delta\omega^\gamma)^\alpha\}&\leq \frac{1}{(1+\epsilon)^\alpha}\Tr\{(\omega^\gamma\rho\omega^\gamma)^\alpha\}+\frac{\epsilon^\alpha}{(1+\epsilon)^\alpha}\Tr\{(\omega^\gamma Q\omega^\gamma)^\alpha\}\label{eq:in-question}\\
&= \frac{1}{(1+\epsilon)^\alpha}2^{(1-\alpha)\tilde{H}_\alpha^\downarrow(A|B)_\rho}+\frac{\epsilon^\alpha}{(1+\epsilon)^\alpha}\Tr\{(\omega^\gamma Q\omega^\gamma)^\alpha\}\\
&\leq \frac{1}{(1+\epsilon)^\alpha}2^{(1-\alpha)\tilde{H}_\alpha^\downarrow(A|B)_\rho}+\frac{\epsilon^\alpha}{(1+\epsilon)^\alpha}2^{(1-\alpha)\tilde{H}_\alpha^\uparrow(A|B)_Q}\\
&\leq \frac{1}{(1+\epsilon)^\alpha}2^{(1-\alpha)\tilde{H}_\alpha^\downarrow(A|B)_\rho}+\frac{\epsilon^\alpha}{(1+\epsilon)^\alpha}|A|^{1-\alpha}\ .\label{eq:upper}
\end{align}
Here the first inequality follows from McCarthy's inequality since $\omega^\gamma\Delta\omega^\gamma=X+Y$. The first equality follows from (\ref{eq:max-exp-rho}). Second inequality follows (\ref{eq:up-exp}). The third inequality follows from the upper bound (\ref{eq:up-log}). Note that (\ref{eq:in-question}) was proved in \cite{marwah2022uniform}, but was applied to different states while proving the continuity inequality for $\tilde{H}^\uparrow_\alpha$.

On the other hand, since the trace functional $\Delta_{AB}\mapsto\Tr\{(\omega^\gamma\Delta_{AB}\omega^\gamma)^\alpha\}$ is concave for $\frac{1}{2}\leq\alpha<1$ \cite{frank2013monotonicity}, we have
\begin{align}
\Tr\{(\omega^\gamma\Delta\omega^\gamma)^\alpha\}&\geq \frac{1}{1+\epsilon}\Tr\{(\omega^\gamma\sigma\omega^\gamma)^\alpha\}+\frac{\epsilon}{1+\epsilon}\Tr\{(\omega^\gamma P \omega^\gamma)^\alpha\}\\
&\geq \frac{1}{1+\epsilon} 2^{(1-\alpha)\tilde{H}_\alpha^\downarrow(A|B)_\sigma}\ .\label{eq:lower}
\end{align}
The second inequality holds from (\ref{eq:max-exp-sigma}) and since $\Tr\{(\omega^\gamma P \omega^\gamma)^\alpha\}\geq 0$.

Thus, combining (\ref{eq:upper}) and (\ref{eq:lower}), we obtain
$$\frac{1}{1+\epsilon} 2^{(1-\alpha)\tilde{H}_\alpha^\downarrow(A|B)_\sigma}\leq   \frac{1}{(1+\epsilon)^\alpha}2^{(1-\alpha)\tilde{H}_\alpha^\downarrow(A|B)_\rho}+\frac{\epsilon^\alpha}{(1+\epsilon)^\alpha}|A|^{1-\alpha}\ .$$
Therefore
\begin{align}
\tilde{H}_\alpha^\downarrow(A|B)_\sigma-\tilde{H}_\alpha^\downarrow(A|B)_\rho&\leq \frac{1}{1-\alpha}\log\{ (1+\epsilon)^{1-\alpha}2^{(1-\alpha)\tilde{H}_\alpha^\downarrow(A|B)_\rho}+(1+\epsilon)^{1-\alpha}\epsilon^\alpha |A|^{1-\alpha}\} - \tilde{H}_\alpha^\downarrow(A|B)_\rho\\
&=\frac{1}{1-\alpha}\log\{ (1+\epsilon)^{1-\alpha}2^{(1-\alpha)\tilde{H}_\alpha^\downarrow(A|B)_\rho}+(1+\epsilon)^{1-\alpha}\epsilon^\alpha |A|^{1-\alpha}\}+\frac{1}{1-\alpha}\log 2^{-(1-\alpha)\tilde{H}_\alpha^\downarrow(A|B)_\rho}\\
&=\frac{1}{1-\alpha}\log\{ (1+\epsilon)^{1-\alpha}+(1+\epsilon)^{1-\alpha}\epsilon^\alpha |A|^{1-\alpha}2^{-(1-\alpha)\tilde{H}_\alpha^\downarrow(A|B)_\rho}\}\\
&\leq\frac{1}{1-\alpha}\log\{ (1+\epsilon)^{1-\alpha}+(1+\epsilon)^{1-\alpha}\epsilon^\alpha |A|^{1-\alpha}|A|^{1-\alpha}\}\\
&=\log(1+\epsilon)+\frac{1}{1-\alpha}\log\Bigg\{1+\epsilon^\alpha |A|^{2(1-\alpha)}\Bigg\}\ .
\end{align}
Here we used the lower dimensional bound (\ref{eq:down-log}): $-\tilde{H}_\alpha^\downarrow(A|B)_\rho\leq \log|A|$.

This inequality also holds with $\rho$ and $\sigma$ interchanged.
\end{proof}

With a similar proof, a continuity inequality can be shown for $\alpha>1$.

\begin{theorem}\label{thm:sand-R-cond-greater}
Let $\alpha>1$. Suppose that the states $\rho_{AB}$ and $\sigma_{AB}$ have the same marginals $\rho_B=\sigma_B$ and they are close to each other in trace-distance $\frac{1}{2}\|\rho-\sigma\|_1=\epsilon\in[0,1]$. Then
\begin{equation}
|\tilde{H}_\alpha^\downarrow(A|B)_\rho-\tilde{H}_\alpha^\downarrow(A|B)_\sigma|\leq \frac{\alpha}{\alpha-1}\log(1+\epsilon)\ .
\end{equation}
Note that this bound is dimension-independent.
\end{theorem}
\begin{proof}

Recall McCarthy's inequality \cite{mccarthycp}: for $X, Y \geq 0$ and $\alpha>1$, we have
\begin{equation}
\Tr \{(X+Y)^\alpha\}\geq \Tr\{X^\alpha\}+\Tr\{Y^\alpha\}\ .
\end{equation}

Following the proof of the previous theorem, take $X=\frac{1}{1+\epsilon}\omega^\gamma\rho\omega^\gamma$ and $Y=\frac{\epsilon}{1+\epsilon}\omega^\gamma Q\omega^\gamma$ in the McCarthy's inequality. Then we have
\begin{align}
\Tr\{(\omega^\gamma\Delta\omega^\gamma)^\alpha\}&\geq \frac{1}{(1+\epsilon)^\alpha}\Tr\{(\omega^\gamma\rho\omega^\gamma)^\alpha\}+\frac{\epsilon^\alpha}{(1+\epsilon)^\alpha}\Tr\{(\omega^\gamma Q\omega^\gamma)^\alpha\}\\
&\geq \frac{1}{(1+\epsilon)^\alpha}2^{(1-\alpha)\tilde{H}_\alpha^\downarrow(A|B)_\rho}\ .\label{eq:upper-greater-one}
\end{align}
Here the first inequality follows from McCarthy's inequality since $\omega^\gamma\Delta\omega^\gamma=X+Y$. The second inequality follows from (\ref{eq:max-exp-rho}) and since $\Tr\{(\omega^\gamma Q \omega^\gamma)^\alpha\}\geq 0$.

On the other hand, since the trace functional $\Delta_{AB}\mapsto\Tr\{(\omega^\gamma\Delta_{AB}\omega^\gamma)^\alpha\}$ is convex for $\alpha>1$ \cite{frank2013monotonicity}, we have
\begin{align}
\Tr\{(\omega^\gamma\Delta\omega^\gamma)^\alpha\}&\leq \frac{1}{1+\epsilon}\Tr\{(\omega^\gamma\sigma\omega^\gamma)^\alpha\}+\frac{\epsilon}{1+\epsilon}\Tr\{(\omega^\gamma P \omega^\gamma)^\alpha\}\\
&= \frac{1}{1+\epsilon} 2^{(1-\alpha)\tilde{H}_\alpha^\downarrow(A|B)_\sigma}+\frac{\epsilon}{1+\epsilon}\Tr\{(\omega^\gamma P \omega^\gamma)^\alpha\}\\
&\leq \frac{1}{1+\epsilon} 2^{(1-\alpha)\tilde{H}_\alpha^\downarrow(A|B)_\sigma}+\frac{\epsilon}{1+\epsilon}2^{(1-\alpha)\tilde{H}^\uparrow_\alpha(A|B)_P}\\
&\leq \frac{1}{1+\epsilon} 2^{(1-\alpha)\tilde{H}_\alpha^\downarrow(A|B)_\sigma}+\frac{\epsilon}{1+\epsilon}|A|^{1-\alpha}\ .\label{eq:lower-greater-one}
\end{align}
The second inequality holds from (\ref{eq:up-exp}), and the last one from (\ref{eq:up-log}).

Thus, combining (\ref{eq:upper-greater-one}) and (\ref{eq:lower-greater-one}), we obtain
$$\frac{1}{(1+\epsilon)^\alpha} 2^{(1-\alpha)\tilde{H}_\alpha^\downarrow(A|B)_\rho}\leq   \frac{1}{1+\epsilon}2^{(1-\alpha)\tilde{H}_\alpha^\downarrow(A|B)_\sigma}+\frac{\epsilon}{1+\epsilon}|A|^{1-\alpha}\ .$$
Since $\alpha>1$,
$$-\tilde{H}_\alpha^\downarrow(A|B)_\rho\leq  \frac{1}{\alpha-1}\log \{(1+\epsilon)^{\alpha-1}2^{(1-\alpha)\tilde{H}_\alpha^\downarrow(A|B)_\sigma}+\epsilon (1+\epsilon)^{\alpha-1}|A|^{1-\alpha}\}\ .$$
Therefore,
\begin{align}
\tilde{H}_\alpha^\downarrow(A|B)_{\sigma}-\tilde{H}_\alpha^\downarrow(A|B)_{\rho}&\leq \tilde{H}_\alpha^\downarrow(A|B)_{\sigma}+\frac{1}{\alpha-1}\log \{(1+\epsilon)^{\alpha-1}2^{(1-\alpha)\tilde{H}_\alpha^\downarrow(A|B)_\sigma}+\epsilon (1+\epsilon)^{\alpha-1}|A|^{1-\alpha}\}\\
&=\frac{1}{\alpha-1}\log 2^{(\alpha-1)\tilde{H}_\alpha^\downarrow(A|B)_\sigma}+\frac{1}{\alpha-1}\log \{(1+\epsilon)^{\alpha-1}2^{(1-\alpha)\tilde{H}_\alpha^\downarrow(A|B)_\sigma}+\epsilon (1+\epsilon)^{\alpha-1}|A|^{1-\alpha}\}\\
&=\frac{1}{\alpha-1}\log\{ (1+\epsilon)^{\alpha-1}+\epsilon(1+\epsilon)^{\alpha-1} |A|^{1-\alpha}2^{(\alpha-1)\tilde{H}_\alpha^\downarrow(A|B)_\sigma}\}\\
&\leq\frac{1}{\alpha-1}\log\{ (1+\epsilon)^{\alpha-1}+\epsilon(1+\epsilon)^{\alpha-1} |A|^{1-\alpha} |A|^{\alpha-1}\}\\
&=\frac{\alpha}{\alpha-1}\log(1+\epsilon)\ .
\end{align}
Here we used the bound (\ref{eq:down-log}): $\tilde{H}_\alpha^\downarrow(A|B)_\sigma\leq \log|A|$.

This inequality also holds with $\rho$ and $\sigma$ interchanged.
\end{proof}

The {\bf $\alpha$-R\'enyi channel entropy} \cite{GW21} is then defined as
\begin{equation}\label{def:entropy-channel-Renyi}
\tilde{S}_\alpha(\cN)=\log|B|-\tilde{D}_\alpha(\cN\|\cR)\ .
\end{equation}
The channel entropy can be written as a infimum of a R\'enyi conditional entropy:
\begin{align}
\tilde{S}_\alpha(\cN)&=\log|B|-\sup_\psi \tilde{D}(\rho_{BR}\|\pi_B\otimes\rho_R)\\
&=-\sup_\psi \tilde{D}_\alpha(\rho_{BR}\|I_B\otimes \rho_R)\\
&=\inf_\psi \tilde{H}_\alpha^\downarrow(B|R)_{\cN(\psi)}\label{eq:Renyi-entropy-conditional}
\end{align}
Here $\rho_{BR}(\psi)=\cN\otimes I \ket{\psi}\bra{\psi}$, and we used that $\tilde{D}_\alpha(\rho\|c\sigma)=\tilde{D}_\alpha(\rho\|\sigma)-\log c$.

Similar to the channel entropy (\ref{def:entropy-channel}), the R\'enyi channel entropy is monotone, normalized, additive, and bounded \cite{GW21}. In particular, because of (\ref{eq:Renyi-entropy-conditional}), the boundedness follows from (\ref{eq:down-log}). Similarly, the lowest value is achieved for the identity channel, and the highest value for the completely randomizing channel.

\begin{theorem}\label{thm:cont-R-entropy} (Continuity of the sandwiched R\'enyi channel entropy) Let $\cN_{A\rightarrow B}$ and $\cM_{A\rightarrow B}$ be two channels from $A$ to $B$ such that $\frac{1}{2}\|\cN-\cM\|_\diamond\leq\epsilon$. Then 
$$|\tilde{S}_\alpha(\cN)-\tilde{S}_\alpha(\cM)|\leq f_{\alpha, |B|}(\epsilon)\ , $$
where 
$$f_{\alpha,|B|}(\epsilon)=\begin{cases}
\log(1+\epsilon)+\frac{1}{1-\alpha}\log\Bigg(1+\epsilon^\alpha |B|^{2(1-\alpha)}\Bigg)\ , \ \ \ \alpha\in[\frac{1}{2}, 1)\\
\frac{\alpha}{\alpha-1}\log(1+\epsilon)\ , \ \ \ \alpha>1 \ .
\end{cases}
$$
\end{theorem}
\begin{proof} Because of the expression of the channel entropy in terms of the conditional entropy (\ref{eq:Renyi-entropy-conditional}), the proof follows the same line of argument as the proof of Theorem \ref{thm:cont}. Since the marginals are the same $\Tr_B(\cN_{A\rightarrow B}\otimes I_R \ket{\omega}\bra{\omega}_{AR})=\Tr_B(\cM_{A\rightarrow B}\otimes I_R \ket{\omega}\bra{\omega}_{AR})$, we use Theorems \ref{thm:sand-R-cond} and \ref{thm:sand-R-cond-greater} to complete the proof.
\end{proof}

\subsection{Sandwiched Tsallis relative entropy}\label{sec:Tsallis-sand}

Tsallis entropy is defined as for $\alpha\in(0,1)\cup(1,\infty)$
\begin{equation}\label{eq:T-entropy}
S^T_\alpha(\rho)=\frac{1}{1-\alpha} (\Tr\rho^\alpha-1)\ .
\end{equation}

The sandwiched Tsallis relative entropy is defined as
\begin{equation}
\tilde{D}^T_\alpha(\rho\|\sigma)=\frac{1}{\alpha-1}\Bigg( \Tr\{(\sigma^{\frac{1-\alpha}{2\alpha}}\rho\sigma^{\frac{1-\alpha}{2\alpha}})^\alpha\}-1\Bigg)\ ,
\end{equation}
for $\alpha<1$, or $\alpha>1$ and $\supp\rho\subseteq\supp\sigma$.

The sandwiched Tsallis relative entropy obeys the data processing inequality for $\alpha\in[\frac{1}{2}, 1)\cup (1,\infty)$ and is jointly convex for $\alpha\in[\frac{1}{2}, 1)$. Moreover, the functional $(\rho,\sigma)\mapsto \Tr\{(\sigma^{\frac{1-\alpha}{2\alpha}}\rho\sigma^{\frac{1-\alpha}{2\alpha}})^\alpha\}$ is jointly concave for $\alpha\in[\frac{1}{2}, 1)$ and jointly convex for $\alpha>1$ \cite{frank2013monotonicity}.

Note that for $d$-dimensional quantum states
\begin{equation}\label{eq:entropy-rel-T}
0\leq S^T_\alpha(\rho)=\frac{d^{1-\alpha}-1}{1-\alpha}-d^{1-\alpha}\tilde{D}_\alpha^T(\rho\|\pi)\leq \frac{d^{1-\alpha}-1}{1-\alpha},
\end{equation}
where $\pi=\frac{I}{d}$.

Conditional sandwiched Tsallis entropies for a state $\rho_{AB}$ are defined as
\begin{align}
\tilde{T}^\downarrow_\alpha(A|B)_\rho&=-\tilde{D}^T_\alpha(\rho_{AB}\|I_A\otimes \rho_B)=\frac{1}{1-\alpha}\Bigg( \Tr\{(\rho_B^{\frac{1-\alpha}{2\alpha}}\rho\rho_B^{\frac{1-\alpha}{2\alpha}})^\alpha\}-1\Bigg)\ ,\\
\tilde{T}^\uparrow_\alpha(A|B)_\rho&=-\min_{\sigma_B}\tilde{D}^T_\alpha(\rho_{AB}\|I_A\otimes \sigma_B)\ .
\end{align}
Since sandwiched Tsallis relative entropy is monotone under partial traces, we have
\begin{align}
\tilde{T}_\alpha^\downarrow(A|B)_\rho=-\tilde{D}^T_\alpha(\rho_{AB}\|I_A\otimes\rho_B)\leq -\tilde{D}^T_\alpha(\rho_A\|I_A)=S^T_\alpha(\rho_A)\leq \frac{|A|^{1-\alpha}-1}{1-\alpha} , \label{eq:T-upper-entropy}
\end{align}
and
\begin{equation}\label{eq:tilde-up-upper}
\tilde{T}^\uparrow_\alpha(A|B)_\rho=-\min_{\sigma_B}\tilde{D}^T_\alpha(\rho_{AB}\|I_A\otimes \sigma_B)\leq -\tilde{D}^T_\alpha(\rho_{A}\|I_A)=S^T_\alpha(\rho_A)\leq \frac{|A|^{1-\alpha}-1}{1-\alpha}\ .
\end{equation}
To show the lower bounds on these conditional entropies, consider the Tsallis relative entropy
$$ D^T_\alpha(\rho\|\sigma)=\frac{1}{\alpha-1}(\Tr\{\rho^\alpha\sigma^{1-\alpha}\}-1)\ . $$

 The conditional Tsallis entropies for a state $\rho_{AB}$ are defined as
\begin{align}
{T}^\downarrow_\alpha(A|B)_\rho&=-{D}^T_\alpha(\rho_{AB}\|I_A\otimes \rho_B)=\frac{1}{1-\alpha}\Bigg( \Tr\{\rho^\alpha \rho_B^{1-\alpha}\}-1\Bigg)\ ,\\
{T}^\uparrow_\alpha(A|B)_\rho&=-\min_{\sigma_B}{D}^T_\alpha(\rho_{AB}\|I_A\otimes \sigma_B)=\frac{1}{1-\alpha}\Bigg(\Big(\Tr\{(\Tr_A\{\rho_{AB}^\alpha\})^\frac{1}{\alpha}\}\Big)^\alpha-1 \Bigg)\ . \label{eq:closed}
\end{align}
The  closed expression for ${T}^\uparrow_\alpha$ is derived similarly to the R\'enyi conditional entropy, as was done in \cite[Lemma 1]{tomamichel2014relating}. Here, however, we will only focus on ${T}^\downarrow_\alpha$. 

Since Tsallis relative entropy is monotone under quantum channels we have the upper bound 
\begin{equation}\label{eq:Tsallis-upperbound}
T_\alpha^\downarrow(A|B)_\rho\leq S^T_\alpha(\rho_A)\leq \frac{|A|^{1-\alpha}-1}{1-\alpha}\ .
\end{equation}
The lower bound on $T_\alpha^\downarrow$ is shown through the duality inequality.  Similarly to the R\'enyi conditional entropy \cite{tomamichel2009fully}, we have the following duality equality:
\begin{proposition}\label{Prop-dual}
Let $\rho_{ABC}$ be a pure state. Then for $\alpha\in(0,2)$, we have
\begin{equation}
T_\alpha^\downarrow(A|B)_\rho+T_{2-\alpha}^\downarrow(A|C)_\rho=0\ .
\end{equation}
\end{proposition}
\begin{proof}
Let $\rho_{ABC}=\ket{\phi}\bra{\phi}_{ABC}$. Then the marginal states satisfy $(\rho_{AB}\otimes I_C)\ket{\phi}=(I_{AB}\otimes\rho_C)\ket{\phi}$ and $(I_{A}\otimes\rho_B\otimes I_C)\ket{\phi}=(\rho_{AC}\otimes I_B)\ket{\phi}$. Therefore,
\begin{align}
(1-\alpha) T_\alpha^\downarrow(A|B)_\rho&=(\alpha-1)D^T_\alpha(\rho_{AB}\|I_A\otimes \rho_B)\\
&=\Tr(\rho_{AB}^\alpha\rho_B^{1-\alpha})-1\\
&=\Tr(\ket{\phi}\bra{\phi}_{ABC}\rho_{AB}^{\alpha-1}\rho_B^{1-\alpha})-1\\
&=\bra{\phi}\rho_{AB}^{\alpha-1}\rho_B^{1-\alpha}\ket{\phi}-1\\
&=\bra{\phi}\rho_{C}^{\alpha-1}\rho_{AC}^{1-\alpha}\ket{\phi}-1\\
&=\Tr(\ket{\phi}\bra{\phi}_{ABC}\rho_{AC}^{1-\alpha}\rho_C^{\alpha-1})-1\\
&=\Tr(\rho_{AC}^{2-\alpha}\rho_C^{1-(2-\alpha)})-1\\
&=(2-\alpha-1)D_{2-\alpha}(\rho_{AC}\|I_A\otimes \rho_C)\\
&=-(1-\alpha) T_{2-\alpha}^\downarrow(A|C)_\rho \ .
\end{align}
\end{proof}
Using this duality relation, we show that all conditional entropies are bounded.
\begin{proposition}
For $\alpha\in(0,2)$ both conditional entropies are bounded
$$-\frac{|A|^{\alpha-1}-1}{\alpha-1}\leq T_\alpha^\downarrow(A|B)_\rho\leq \frac{|A|^{1-\alpha}-1}{1-\alpha}\ ,$$
\begin{equation}\label{eq:T-tilde-down-bound}
-\frac{|A|^{\alpha-1}-1}{\alpha-1}\leq \tilde{T}_\alpha^\downarrow(A|B)_\rho\leq \frac{|A|^{1-\alpha}-1}{1-\alpha}\ . 
\end{equation}
\end{proposition}
\begin{proof}
Tsallis conditional entropies are upper bounded by the arguments above (\ref{eq:T-upper-entropy}), (\ref{eq:Tsallis-upperbound}).
To show the lower bound on the Tsallis conditional entropy, let us take a purification $\ket{\phi}\bra{\phi}_{ABC}$ of $\rho_{AB}$. Then from the Proposition  \ref{Prop-dual} and (\ref{eq:Tsallis-upperbound}), we have
$$T_\alpha^\downarrow(A|B)_\rho= T_\alpha^\downarrow(A|B)_\phi=-T_{2-\alpha}^\downarrow(A|C)_\phi\geq -S^T_{2-\alpha}(\rho_A)\geq -\frac{|A|^{\alpha-1}-1}{\alpha-1}\ .$$

Also, for $\alpha\in[0,\infty]$, we have the relation $\tilde{D}^T_\alpha(\rho)\leq D^T_\alpha(\rho)$, \cite{araki1990inequality, hiai1994equality, lieb2001inequalities} , and therefore,
$$T_\alpha^\downarrow(A|B)_\rho\leq \tilde{T}_\alpha^\downarrow(A|B)_\rho \ . $$

\end{proof}

Similarly to the sandwiched R\'enyi conditional entropy we have the following continuity inequality.

\begin{theorem}\label{thm:Tsallis-cont}
Let $\alpha\in[\frac{1}{2}, 1)$. Suppose that the states $\rho_{AB}$ and $\sigma_{AB}$ have the same marginals $\rho_B=\sigma_B$ and they are close to each other in trace-distance $\frac{1}{2}\|\rho-\sigma\|_1=\epsilon\in[0,1]$. Then
\begin{equation}
|\tilde{T}_\alpha^\downarrow(A|B)_\rho-\tilde{T}_\alpha^\downarrow(A|B)_\sigma|\leq  \frac{1}{1-\alpha}((1+\epsilon^\alpha )(1+\epsilon)^{1-\alpha}-1)|A|^{1-\alpha}\ .
\end{equation}
\end{theorem}

\begin{proof}
The proof is similar to the proof of Theorem \ref{thm:sand-R-cond} with a few differences. Equalities (\ref{eq:max-exp-rho}, \ref{eq:max-exp-sigma}) are replaced with
\begin{equation}\label{eq:T-down-trace}
\Tr\{(\omega^\gamma\rho\omega^\gamma)^\alpha\}=(1-\alpha)\tilde{T}^\downarrow_\alpha(A|B)_\rho+1\ , \ \ \Tr\{(\omega^\gamma\sigma\omega^\gamma)^\alpha\}=(1-\alpha)\tilde{T}^\downarrow_\alpha(A|B)_\sigma+1\ .
\end{equation}
And inequality (\ref{eq:up-exp}) is replaced with, for any state $\delta_{AB}$,
\begin{equation}\label{eq:T-delta-bound}
\Tr\{ (\omega^\gamma\delta\omega^\gamma)^\alpha\}\leq (1-\alpha)\tilde{T}^\uparrow_\alpha(A|B)_\delta+1\ .
\end{equation}
And therefore, the upper bound (\ref{eq:upper}) becomes
\begin{align}
\Tr\{(\omega^\gamma\Delta\omega^\gamma)^\alpha\}&\leq \frac{1}{(1+\epsilon)^\alpha}\Bigg((1-\alpha)\tilde{T}_\alpha^\downarrow(A|B)_\rho+1\Bigg)+\frac{\epsilon^\alpha}{(1+\epsilon)^\alpha}|A|^{1-\alpha}\ .
\end{align}
Here we either used the upper bound $\tilde{T}_\alpha^\uparrow(A|B)_Q\leq \frac{|A|^{1-\alpha}-1}{1-\alpha}$ (\ref{eq:tilde-up-upper}) or the fact that we showed that $\Tr\{(\omega^\gamma Q\omega^\gamma)^\alpha\}\leq |A|^{1-\alpha}$.

The lower bound (\ref{eq:lower}) becomes
\begin{equation}
\Tr\{(\omega^\gamma\Delta\omega^\gamma)^\alpha\}\geq\frac{1}{1+\epsilon} \Bigg( (1-\alpha)\tilde{T}_\alpha^\downarrow(A|B)_\sigma+1\Bigg)\ .
\end{equation}
Combining the last two inequalities, we obtain
\begin{equation}
(1-\alpha)\tilde{T}_\alpha^\downarrow(A|B)_\sigma+1 \leq (1+\epsilon)^{1-\alpha}\Bigg((1-\alpha)\tilde{T}_\alpha^\downarrow(A|B)_\rho+1\Bigg)+\epsilon^\alpha (1+\epsilon)^{1-\alpha}|A|^{1-\alpha}\ .
\end{equation}
Therefore,
\begin{align}
\tilde{T}_\alpha^\downarrow(A|B)_\sigma-\tilde{T}_\alpha^\downarrow(A|B)_\rho&\leq  (1+\epsilon)^{1-\alpha}\Bigg(\tilde{T}_\alpha^\downarrow(A|B)_\rho+\frac{1}{1-\alpha} \Bigg)+\frac{1}{1-\alpha}\epsilon^\alpha (1+\epsilon)^{1-\alpha}|A|^{1-\alpha}-\frac{1}{1-\alpha}-\tilde{T}_\alpha^\downarrow(A|B)_\rho\\
&= ((1+\epsilon)^{1-\alpha}-1)\tilde{T}_\alpha^\downarrow(A|B)_\rho+ \frac{1}{1-\alpha}\epsilon^\alpha (1+\epsilon)^{1-\alpha}|A|^{1-\alpha}+\frac{1}{1-\alpha} ((1+\epsilon)^{1-\alpha}-1)\\
&\leq ((1+\epsilon)^{1-\alpha}-1)\frac{|A|^{1-\alpha}-1}{1-\alpha}+ \frac{1}{1-\alpha}\epsilon^\alpha (1+\epsilon)^{1-\alpha}|A|^{1-\alpha}+\frac{1}{1-\alpha} ((1+\epsilon)^{1-\alpha}-1)\\
&\leq ((1+\epsilon)^{1-\alpha}-1)\frac{|A|^{1-\alpha}}{1-\alpha}+\frac{1}{1-\alpha} \epsilon^\alpha (1+\epsilon)^{1-\alpha}|A|^{1-\alpha}\\
&=\frac{1}{1-\alpha} ((1+\epsilon^\alpha) (1+\epsilon)^{1-\alpha}-1)|A|^{1-\alpha}\ .
\end{align}
Here we used the upper dimensional bound $\tilde{T}_\alpha^\downarrow(A|B)_\rho\leq \frac{|A|^{1-\alpha}-1}{1-\alpha}$ (\ref{eq:T-upper-entropy}), since $(1+\epsilon)^{1-\alpha}-1>0$ for $\alpha\in[\frac{1}{2}, 1)$.
\end{proof}

For $\alpha>1$, we have the following continuity inequality.
\begin{theorem}\label{thm:Tsallis-cont-greater}
Let $\alpha\in(1,2)$. Suppose that the states $\rho_{AB}$ and $\sigma_{AB}$ have the same marginals $\rho_B=\sigma_B$ and they are close to each other in trace-distance $\frac{1}{2}\|\rho-\sigma\|_1=\epsilon\in[0,1]$. Then
\begin{equation}
|\tilde{T}_\alpha^\downarrow(A|B)_\rho-\tilde{T}_\alpha^\downarrow(A|B)_\sigma|\leq  \frac{1}{\alpha-1} \Bigg\{\Bigg((1+\epsilon)^{\alpha-1}-1\Bigg)|A|^{\alpha-1}+\epsilon(1+\epsilon)^{\alpha-1}|A|^{1-\alpha}\Bigg\}\ .
\end{equation}
\end{theorem}

\begin{proof}
Using expressions (\ref{eq:T-down-trace}) instead of (\ref{eq:max-exp-rho}, \ref{eq:max-exp-sigma}) in the proof of Theorem \ref{thm:sand-R-cond-greater}, instead of (\ref{eq:upper-greater-one}), we obtain
\begin{equation}
\Tr\{(\omega^\gamma\Delta\omega^\gamma)^\alpha\}\geq\frac{1}{(1+\epsilon)^\alpha} \Bigg((1-\alpha)\tilde{T}_\alpha^\downarrow(A|B)_\rho +1\Bigg) \ .
\end{equation}
Instead of (\ref{eq:lower-greater-one}) we obtain
\begin{equation}
\Tr\{(\omega^\gamma\Delta\omega^\gamma)^\alpha\}\leq \frac{1}{1+\epsilon} \Bigg((1-\alpha)\tilde{T}_\alpha^\downarrow(A|B)_\sigma+1 \Bigg)+\frac{\epsilon}{1+\epsilon}|A|^{1-\alpha}\ .
\end{equation}
Thus, combining the last two inequalities, since $\alpha>1$, we have
$$ 
\tilde{T}_\alpha^\downarrow(A|B)_\rho\geq \frac{1}{1-\alpha}\Bigg\{(1+\epsilon)^{\alpha-1}\Bigg((1-\alpha)\tilde{T}_\alpha^\downarrow(A|B)_\sigma+1\Bigg)+\epsilon(1+\epsilon)^{\alpha-1}|A|^{1-\alpha}  -1\Bigg\}\ .
$$
Therefore, 
\begin{align}
\tilde{T}_\alpha^\downarrow(A|B)_\sigma-\tilde{T}_\alpha^\downarrow(A|B)_\rho&\leq \tilde{T}_\alpha^\downarrow(A|B)_\sigma+\frac{1}{\alpha-1}\Bigg\{(1+\epsilon)^{\alpha-1}\Bigg((1-\alpha)\tilde{T}_\alpha^\downarrow(A|B)_\sigma+1\Bigg)+\epsilon(1+\epsilon)^{\alpha-1}|A|^{1-\alpha}  -1\Bigg\}\\
&=\tilde{T}_\alpha^\downarrow(A|B)_\sigma\Bigg(1-(1+\epsilon)^{\alpha-1} \Bigg)+\frac{1}{\alpha-1}\Bigg\{(1+\epsilon)^{\alpha-1}+\epsilon(1+\epsilon)^{\alpha-1}|A|^{1-\alpha}  -1\Bigg\}\\
&\leq \frac{|A|^{\alpha-1}-1}{\alpha-1}\Bigg((1+\epsilon)^{\alpha-1}-1 \Bigg)+\frac{1}{\alpha-1}\Bigg\{(1+\epsilon)^{\alpha-1}+\epsilon(1+\epsilon)^{\alpha-1}|A|^{1-\alpha}  -1\Bigg\}\\
&= \frac{1}{\alpha-1} \Bigg\{ \Bigg(|A|^{\alpha-1}-1\Bigg)\Bigg((1+\epsilon)^{\alpha-1}-1 \Bigg)+(1+\epsilon)^{\alpha-1}+\epsilon(1+\epsilon)^{\alpha-1}|A|^{1-\alpha}  -1\Bigg\}\\
&= \frac{1}{\alpha-1} \Bigg\{\Bigg((1+\epsilon)^{\alpha-1}-1\Bigg)|A|^{\alpha-1}+\epsilon(1+\epsilon)^{\alpha-1}|A|^{1-\alpha}\Bigg\}\\
\end{align}
Here, since $1-(1+\epsilon)^{\alpha-1}<0$ for $\alpha>1$, we used the bound (\ref{eq:T-tilde-down-bound}): $-\tilde{T}_\alpha^\downarrow(A|B)_\sigma\leq \frac{|A|^{\alpha-1}-1}{\alpha-1}$.
\end{proof}

Note that the entropy of a state (\ref{eq:entropy-rel-T}) is related to the Tsallis sandwiched relative entropy as
$$S_\alpha^T(\rho)=\frac{|B|^{1-\alpha}-1}{1-\alpha}-|B|^{1-\alpha}\tilde{D}^T_\alpha(\rho\|\pi)\ , $$
where $\pi=I/|B|$.

Similarly, the {\bf $\alpha$-Tsallis channel entropy} is defined as
\begin{equation}\label{def:entropy-channel-Tsallis}
\tilde{S}^T_\alpha(\cN)=\frac{|B|^{1-\alpha}-1}{1-\alpha}-|B|^{1-\alpha}\tilde{D}^T_\alpha(\cN\|\cR)\ .
\end{equation}
The channel entropy can be written as a infimum of a Tsallis conditional entropy:
\begin{align}
\tilde{S}^T_\alpha(\cN)&=\frac{|B|^{1-\alpha}-1}{1-\alpha}-|B|^{1-\alpha}\sup_\psi\tilde{D}^T_\alpha(\rho_{BR}\|\pi_B\otimes\rho_R)\\
&=-\sup_\psi \tilde{D}^T_\alpha(\rho_{BR}\|I_B\otimes \rho_R)\\
&=\inf_\psi \tilde{T}_\alpha^\downarrow(B|R)_{\cN(\psi)}\label{eq:Tsallis-entropy-conditional}
\end{align}
Here $\rho_{BR}(\psi)=\cN\otimes I \ket{\psi}\bra{\psi}$, and we used that $\tilde{D}^T_\alpha(\rho\|c\sigma)=\frac{c^{1-\alpha}-1}{\alpha-1}+c^{1-\alpha}\tilde{D}^T_\alpha(\rho\|\sigma)$.

Similarly to the channel entropy (\ref{def:entropy-channel}) and R\'enyi channel entropy (\ref{def:entropy-channel-Renyi}), the Tsallis channel entropy is monotone under the uniformity preserving superchannels and it is normalized. 

(Normalization) By definition, the entropy of a completely randomizing channel $\cR$ is $\tilde{S}^T_\alpha(\cR)=\frac{|B|^{1-\alpha}-1}{1-\alpha}$. And the entropy of the replacer channel is $\tilde{S}_\alpha^T(\Phi_\sigma)=S_\alpha^T(\sigma)$. Therefore, for a replacer channel that replaces any state with a pure state, the entropy of this channel is zero,  i.e. $\tilde{S}_\alpha^T(\Phi_\phi)=0$ for $\Phi_\phi(\rho)=\ket{\phi}\bra{\phi}$ for some fixed pure state $\ket{\phi}$.

From the bound on the conditional entropy (\ref{eq:T-tilde-down-bound}), the Tsallis channel entropy is also bounded. The upper bound is reached for the completely randomizing channel, and the lower bound is reached for the identity channel.

\begin{theorem}\label{thm:Tsallis-bound} (Boundedness) Let $\alpha\in(0,2)$. The $\alpha$-Tsallis channel entropy is bounded
$$-\frac{|B|^{\alpha-1}-1}{\alpha-1}\leq\tilde{S}^T_\alpha(\cN)\leq \frac{|B|^{1-\alpha}-1}{1-\alpha}\ .$$
\end{theorem}
Note that the Tsallis entropy $S^T_\alpha$  is upper bounded with a bound that is $\alpha$-dependent (\ref{eq:entropy-rel-T}), resulting in different lower and upper bound on the Tsallis channel entropy.

Tsallis relative entropy is pseudo-additive:
$$\tilde{D}_\alpha^T(\rho_1\otimes\rho_2\|\sigma_1\otimes\sigma_2)=\tilde{D}_\alpha^T(\rho_1\|\sigma_1)+\tilde{D}_\alpha^T(\rho_2\|\sigma_2)+(\alpha-1)\tilde{D}_\alpha^T(\rho_1\|\sigma_1)\tilde{D}_\alpha^T(\rho_2\|\sigma_2)\ . $$
Therefore, the channel entropy is pseudo-additive.

\begin{theorem}\label{thm:Tsallis-additive} (Pseudo-additivity)
Let $\alpha>1$, and let $\cN_{A_1\rightarrow B_1}$ and $\cM_{A_2\rightarrow B_2}$ be two channels. Then
$$\tilde{S}_\alpha^T(\cN\otimes\cM)=\tilde{S}_\alpha^T(\cN)+\tilde{S}_\alpha^T(\cM)+(1-\alpha)\tilde{S}_\alpha^T(\cN)\tilde{S}_\alpha^T(\cM)\ .$$
\end{theorem}
\begin{proof}

For completely randomizing channels $\cR_1=\cR_{A_1\rightarrow B_1}$ and $\cR_2=\cR_{A_2\rightarrow B_2}$, we have
$$
\tilde{S}_\alpha^T(\cN\otimes\cM)=\frac{|B_1|^{1-\alpha}|B_2|^{1-\alpha}-1}{1-\alpha}-|B_1|^{1-\alpha}|B_2|^{1-\alpha}\tilde{D}^T_\alpha(\cN\otimes\cM\|\cR_1\otimes\cR_2)\ .$$
Therefore, the equality follows if 
$$\tilde{D}^T_\alpha(\cN\otimes\cM\|\cR_1\otimes\cR_2)=\tilde{D}^T_\alpha(\cN\|\cR_1)+\tilde{D}^T_\alpha(\cM\|\cR_2)+(\alpha-1)\tilde{D}^T_\alpha(\cN\|\cR_1)\tilde{D}^T_\alpha(\cM\|\cR_2) \ . $$
The $"\geq"$ inequality follows directly from the definition of the relative entropy between channels and the pseudo-additivity of the relative entropy between states. The $"\leq"$ inequality follows from the proof of additivity for the R\'enyi entropy \cite{GW21}. We adapt this argument to the Tsallis case by using the relation between $\tilde{D}_\alpha^T$ and $\tilde{D}_\alpha$.

Let $\psi_{RA_1A_2}$ be an arbitrary pure state. Define $\rho_{A_1R'}=\cM(\psi_{RA_1A_2})$ and $\sigma_{A_1R'}=\cR_2(\psi_{RA_1A_2})$, where $R'=B_2R$. Then $\cN\otimes\cM(\psi_{RA_1A_2})=\cN(\rho_{A_1R'})$ and $\cR_1\otimes\cR_2(\psi_{RA_1A_2})=\cR_1(\sigma_{A_1R'})$. Thus,
\begin{align}
\tilde{D}_\alpha^T(\cN\otimes\cM(\psi_{RA_1A_2})\|\cR_1\otimes\cR_2(\psi_{RA_1A_2}))&=\tilde{D}_\alpha^T(\cN(\rho_{A_1R'})\|\cR_1(\sigma_{A_1R'}))=\frac{1}{\alpha-1}\Tr\{(\cR_1(\sigma)^\gamma\cN(\rho)\cR_1(\sigma)^\gamma)^\alpha \}-\frac{1}{\alpha-1}\ . 
\end{align}

Note that the R\'enyi channel entropy is additive for $\alpha>1$, as was discussed in Proposition 15 in  \cite{GW21}. The proof of the additivity relies on the inequality presented in the proof of Proposition 41 in \cite{wilde2020amortized}:
\begin{align}
\tilde{D}_\alpha(\cN(\rho_{A_1R'})\|\cR_1(\sigma_{A_1R'}))\leq \tilde{D}_\alpha(\cN\|\cR_1)+\tilde{D}_\alpha(\rho_{A_1R'}\|\sigma_{A_1R'})\ .
\end{align}
This inequality is equivalent to
\begin{align}\label{eq:X}
\frac{1}{\alpha-1}\log X(\cN(\rho_{A_1R'})\|\cR_1(\sigma_{A_1R'}))\leq \sup_{\xi_{A_1\tilde{R}}} \frac{1}{\alpha-1}\log X(\cN(\xi_{A_1\tilde{R}})\|\cR_1(\xi_{A_1\tilde{R}}))+\frac{1}{\alpha-1}\log X(\rho_{A_1R'}\|\sigma_{A_1R'})\ ,
\end{align}
where $X(\rho\|\sigma)=\Tr\{(\sigma^\gamma\rho \sigma^\gamma)^\alpha \}$, therefore $\tilde{D}_\alpha(\rho\|\sigma)=\frac{1}{\alpha-1}\log X(\rho\|\sigma)$ and $\tilde{D}^T_\alpha(\rho\|\sigma)=\frac{1}{\alpha-1}\bigg( X(\rho\|\sigma)-1\bigg)$. Now, (\ref{eq:X}) is equivalent to
\begin{align}
\frac{1}{\alpha-1} X(\cN(\rho_{A_1R'})\|\cR_1(\sigma_{A_1R'}))\leq \sup_{\xi_{A_1\tilde{R}}}\frac{1}{\alpha-1} X(\cN(\xi_{A_1\tilde{R}})\|\cR_1(\xi_{A_1\tilde{R}}))\cdot X(\rho_{A_1R'}\|\sigma_{A_1R'})\ .
\end{align}

Therefore, applying this result to the Tsallis relative entropy, we have
\begin{align}
\tilde{D}_\alpha^T(\cN\otimes\cM(\psi_{RA_1A_2})\|\cR_1\otimes\cR_2(\psi_{RA_1A_2}))&=\tilde{D}_\alpha^T(\cN(\rho_{A_1R'})\|\cR_1(\sigma_{A_1R'}))\\
&=\frac{1}{\alpha-1}X(\cN(\rho_{A_1R'})\|\cR_1(\sigma_{A_1R'}))\ -\frac{1}{\alpha-1}\\
&\leq \sup_{\xi_{A_1\tilde{R}}}\frac{1}{\alpha-1} X(\cN(\xi_{A_1\tilde{R}})\|\cR_1(\xi_{A_1\tilde{R}}))\cdot X(\rho_{A_1R'}\|\sigma_{A_1R'})-\frac{1}{\alpha-1}\\
&= \sup_{\xi_{A_1\tilde{R}}}\frac{1}{\alpha-1} X(\cN(\xi_{A_1\tilde{R}})\|\cR_1(\xi_{A_1\tilde{R}}))\cdot X(\cM(\psi_{RA_1A_2})\|\cR_2(\psi_{RA_1A_2}))-\frac{1}{\alpha-1}\\
&=\sup_{\xi_{A_1\tilde{R}}} \tilde{D}^T_\alpha(\cN(\xi_{A_1\tilde{R}})\|\cR_1(\xi_{A_1\tilde{R}}))+\tilde{D}^T_\alpha(\cM(\psi_{RA_1A_2})\|\cR_2(\psi_{RA_1A_2}))\\
&\ \ +(\alpha-1)\sup_{\xi_{A_1\tilde{R}}}\tilde{D}^T_\alpha(\cN(\xi_{A_1\tilde{R}})\|\cR_1(\xi_{A_1\tilde{R}}))\cdot\tilde{D}^T_\alpha(\cM(\psi_{RA_1A_2})\|\cR_2(\psi_{RA_1A_2}))\\
&= \tilde{D}^T_\alpha(\cN\|\cR_1)+\tilde{D}^T_\alpha(\cM(\psi_{RA_1A_2})\|\cR_2(\psi_{RA_1A_2}))\\
&\ \ +(\alpha-1)\tilde{D}^T_\alpha(\cN\|\cR_1)\tilde{D}^T_\alpha(\cM(\psi_{RA_1A_2})\|\cR_2(\psi_{RA_1A_2})) \ .
\end{align}
Taking supremum over all states $\psi_{RA_1A_2}$ on both sides, we reach the necessary inequality.
\end{proof}

\begin{theorem}\label{thm:Tsallis-continuity-entropy} (Continuity of the sandwiched Tsallis channel entropy) Let $\cN_{A\rightarrow B}$ and $\cM_{A\rightarrow B}$ be two channels from $A$ to $B$ such that $\frac{1}{2}\|\cN-\cM\|_\diamond\leq\epsilon$. Then
$$|\tilde{S}^T_\alpha(\cN)-\tilde{S}^T_\alpha(\cM)|\leq f^T_{\alpha, |B|}(\epsilon)\ . $$
Here 
$$f^T_{\alpha, d}(\epsilon)=\begin{cases}
 \frac{1}{1-\alpha} \Bigg((1+\epsilon^\alpha) (1+\epsilon)^{1-\alpha}-1\Bigg)d^{1-\alpha}\ , \ \ \ \alpha\in[\frac{1}{2}, 1)\\
 \frac{1}{\alpha-1} \Bigg(((1+\epsilon)^{\alpha-1}-1)d^{\alpha-1}+\epsilon(1+\epsilon)^{\alpha-1}d^{1-\alpha}\Bigg)\ , \ \ \ \alpha\in(1,2) \ .
 \end{cases}
$$
\end{theorem}
\begin{proof} Because of the expression of the channel entropy in terms of the conditional entropy (\ref{eq:Tsallis-entropy-conditional}), the proof follows the same line of argument as the proof of Theorem \ref{thm:cont}. Since the marginals are the same $\Tr_B(\cN_{A\rightarrow B}\otimes I_R \ket{\omega}\bra{\omega}_{AR})=\Tr_B(\cM_{A\rightarrow B}\otimes I_R \ket{\omega}\bra{\omega}_{AR})$, we use Theorems \ref{thm:Tsallis-cont} and \ref{thm:Tsallis-cont-greater} to complete the proof.

\end{proof}

\section{Conclusion}

We proved uniform continuity bounds for the sandwiched R\'enyi and Tsallis conditional entropies  $\tilde{H}_\alpha^\downarrow, \tilde{T}_\alpha^\downarrow$ for states with the same marginal on the conditioning system. The bound depends only on the dimension of the conditioning system, except for the bound of $\tilde{H}_\alpha^\downarrow$ for $\alpha>1$, where the bound is independent of any dimension. We applied these bounds to show that the R\'enyi and Tsallis channel entropies defined through the corresponding sandwiched entropies are continuous with respect to the diamond distance on the channels. Note that we did not consider channel entropies defined through the regular (non-sandwiched) relative entropies, as it is not clear whether these channel entropies are additive. However, it would be of a separate mathematical interest to derive continuity inequalities for the non-sandwiched conditional entropies.

Also note, that looking at the definitions of channel entropies (\ref{def:entropy-channel}), (\ref{def:entropy-channel-Renyi}), (\ref{def:entropy-channel-Tsallis}), it is clear that generalizing channel entropy by using a generalized divergence $D(\cdot\|\cdot)$ to produce a meaningful definition of a channel entropy one must have $S_D(\cN)=f(D(\cN\|\cR)),$ where $D(\rho\| c\sigma)=-f(D(\rho\|\sigma)).$

\vspace{0.3in}
\textbf{Acknowledgments.}  A. V. is supported by NSF grant DMS-2105583.
\vspace{0.3in}

\bibliography{Bibliography-Vershynina}{}
\bibliographystyle{plain}

\end{document}